\documentclass{llncs}

\usepackage{graphicx}
\usepackage{amsmath}
\usepackage{amsfonts}
\usepackage{url}

\usepackage[ruled,vlined,linesnumbered]{algorithm2e}

\usepackage{rotating}

\usepackage{booktabs}

\usepackage{bm}

\usepackage{qtree}

\usepackage{lhelp}

\def\R {\ensuremath{\mathbb{R}}}
\def\Q {\ensuremath{\mathbb{Q}}}

\def\C {\ensuremath{\mathbb{C}}}
\def\x {\ensuremath{\mathbf{x}}}

\spnewtheorem{lem}[theorem]{Lemma}{\bfseries}{\itshape}
\spnewtheorem{prop}[theorem]{Proposition}{\bfseries}{\itshape}

\begin{document}

\title{Truth Table Invariant Cylindrical Algebraic Decomposition by Regular Chains}

\author{
Russell Bradford\inst{1},
Changbo Chen\inst{2},
James H. Davenport\inst{1},
Matthew England\inst{1},
Marc Moreno Maza\inst{3} 
and David Wilson\inst{1}
}

\institute{
$^1$  University of Bath, Bath, BA2 7AY, UK. \\
$^2$  CIGIT, Chinese Academy of Sciences, Chongqing, 400714, China. \\
$^3$  University of Western Ontario, London, Ontario, N6A 5B7, Canada. 
\\ 
	\email{ 
		{\tt \{R.Bradford, J.H.Davenport, M.England, D.J.Wilson\}@bath.ac.uk},\\ 
		{\tt moreno@csd.uwo.ca, changbo.chen@hotmail.com }
	}
}

\maketitle

\begin{abstract}
A new algorithm to compute cylindrical algebraic decompositions (CADs) is presented, building on two recent advances.  Firstly, the output is truth table invariant (a TTICAD) meaning given formulae have constant truth value on each cell of the decomposition.  Secondly, the computation uses regular chains theory to first build a cylindrical decomposition of complex space (CCD) incrementally by polynomial. Significant modification of the regular chains technology was used to achieve the more sophisticated invariance criteria.  Experimental results on an implementation in the \texttt{RegularChains} Library for {\sc Maple} verify that combining these advances gives an algorithm superior to its individual components and competitive with the state of the art.

\keywords{cylindrical algebraic decomposition; equational constraint; regular chains; triangular decomposition}
\end{abstract}

\section{Introduction}
\label{SEC:Introduction}

A \emph{cylindrical algebraic decomposition} (CAD) is a collection of cells such that: they do not intersect and their union describes all of $\R^n$; they are arranged \emph{cylindrically}, meaning the projections of any pair of cells are either equal or disjoint; and, each can be described using a finite sequence of polynomial relations.
  
CAD was introduced by Collins in \cite{Collins1975} to solve quantifier elimination problems, 
and this remains an important application (see \cite{CM14-qe} for details on how our work can be used there). 
Other applications  include epidemic modelling \cite{BENW06}, parametric optimisation \cite{FPM05}, theorem proving \cite{Paulson2012}, robot motion planning \cite{SS83II} and reasoning with multi-valued functions and their branch cuts \cite{DBEW12}.  
CAD has complexity doubly exponential in the number of variables.  While for some applications there now exist algorithms with better complexity (see for example \cite{BPR96}), CAD implementations remain the best general purpose approach for many.

We present a new CAD algorithm combining two recent advances: the technique of producing CADs via regular chains in complex space \cite{CMXY09}, and the idea of producing CADs closely aligned to the structure of logical formulae \cite{BDEMW13}.  We continue by reminding the reader of CAD theory and these advances.  

\subsection{Background on CAD}
\label{SUBSEC:Background}

We work with polynomials in ordered variables $\bm{x} = x_1\prec \ldots \prec x_n$.  
The \emph{main variable} of a polynomial (${\rm mvar}$) is the greatest variable present with respect to the ordering.  Denote by QFF a \emph{quantifier free Tarski formula}: a Boolean combination ($\land,\lor,\neg$) of statements 
$f_i \, \sigma \, 0$ where $\sigma \in \{=,>,<\}$ and the $f_i$ are polynomials.
CAD was developed as a tool for the problem of quantifier elimination over the reals: given a quantified Tarski formula 
\\
\hspace*{1in} $\Psi(x_1, \ldots, x_k) := Q_{k+1}x_{k+1}\ldots Q_nx_n F(x_1,\ldots,x_n)$ 
\\
(where $Q_i\in\{\forall,\exists\}$ and $F$ is a QFF), produce an equivalent QFF $\psi(x_1,\ldots,x_k)$. Collins proposed to build a CAD of $\R^n$ which is \emph{sign-invariant}, so each $f_i \in F$ is either positive, negative or zero on each cell. Then $\psi$ is the disjunction of the defining formulae of those cells $c \in \R^k$ where $\Psi$ is true, which given sign-invariance, requires us to only test one \emph{sample point} per cell.

Collins' algorithm works by first \emph{projecting} the problem into decreasing real dimensions and then \emph{lifting} to build CADs of increasing dimension. 
Important developments range from improved projection operators \cite{McCallum1999} to the use of certified numerics when lifting \cite{Strzebonski06} \cite{IYAY09}.  See for example \cite{BDEMW13} for a fuller discussion.

\subsection{Truth table invariant CAD}
\label{SUBSEC:BathWork}

One important development is the use of \emph{equational constraints} (ECs), which are equations logically implied by a formula.  These may be given explicitly as in $(f=0)\land \varphi$, or implicitly as $f_1f_2=0$ is by $(f_1 = 0 \land \varphi_1) \lor (f_2 = 0 \land \varphi_2)$.  

In \cite{McCallum1999} McCallum developed the theory of a reduced operator for the first projection, so that the CAD produced was sign-invariant for the polynomial defining a given EC, and then sign-invariant for other polynomials only when the EC is satisfied.  
Extensions of this to make use of more than one EC have been investigated (see for example \cite{BM05}) while in \cite{BDEMW14} it was shown how McCallum's theory could allow for further savings in the lifting phase.

The CADs produced are no longer sign-invariant for polynomials but instead \emph{truth-invariant} for a formula.  Truth-invariance was defined in \cite{Brown1998} where sign-invariant CADs were refined to maintain it.  We consider a related definition.
\begin{definition}[\cite{BDEMW13}]
Let $\Phi = \{ \phi_i\}_{i=1}^t$ be a list of QFFs.  A CAD is \emph{Truth Table Invariant} for $\Phi$ (a TTICAD) if on each cell every $\phi_i$ has constant Boolean value.
\end{definition}
In \cite{BDEMW13} an algorithm to build TTICADs when each $\phi_i$ has an EC was derived by extending \cite{McCallum1999} (which could itself apply in this case but would be less efficient).
Implementations in {\sc Maple} showed this offered great savings in both CAD size and computation time when compared to the sign-invariant theory. 
In \cite{BDEMW14} this theory has been extended to work on arbitrary $\phi_i$, with savings if at least one has an EC.
Note that there are two distinct reasons to build a TTICAD:
\begin{enumerateshort}
\item \emph{As a tool to build a truth-invariant CAD:}  If a parent formula $\phi^{*}$ is built from $\{\phi_i\}$ then any TTICAD for $\{\phi_i\}$ is also truth-invariant for $\phi^{*}$.  

A TTICAD may be the best truth-invariant CAD, or at least the best we can compute.  
Note that the TTICAD theory allows for more savings than the use of \cite{McCallum1999} with an implicit EC built as the product of ECs from $\phi_i$ \cite{BDEMW13}.

\item \emph{When truth table invariance is required:}  There are applications which provide a list of formulae but no parent formula.  For example, decomposing complex space according to a set of branch cuts for the purpose of algebraic simplification \cite{BD02} \cite{PBD10} \cite{EBDW13}.  When the branch cuts can be expressed as semi-algebraic systems a TTICAD provides exactly the required decomposition. 
\end{enumerateshort}

\subsection{CAD by regular chains}
\label{SUBSEC:UWOWork}

Recently, a radically different method to build CADs has been investigated.  
Instead of projecting and lifting, the problem is moved to complex space where the theory of triangular decomposition by regular chains is used to build a \emph{complex cylindrical decomposition} (CCD): a decomposition of $\C^n$ such that each cell is cylindrical.  
This is encoded as a tree data structure, with each path through the tree describing the end leaf as a solution of a regular system \cite{Wang2000}.

This was first proposed in \cite{CMXY09} to build a sign-invariant CAD.  
Techniques developed for comprehensive triangular decomposition \cite{CGLMP07} were used to build a sign-invariant decomposition of $\C^n$ which was then refined to a CCD.  
Finally, real root isolation is applied to refine further to a CAD of $\R^n$.  
The computation of the CCD may be viewed as an enhanced projection phase since gcds of pairs of polynomials are calculated as well as resultants.  
The extra work used here makes the second phase, which may be compared to lifting, less expensive.  
The main advantage is the use of case distinction in the second phase, so that the zeros of polynomials not relevant in a particular branch are not isolated there.

The construction of the CCD was improved in \cite{CM12b}.  
The former approach built a decomposition for the input in one step using existing algorithms.  
The latter approach proceeds incrementally by polynomial, each time using purpose-built algorithms to refine an existing tree whilst maintaining cylindricity.  
Experimental results showed that the latter approach is much quicker, with its implementation in
\textsc{Maple}'s \texttt{RegularChains} library now competing with existing state of the art CAD algorithms: \textsc{Qepcad} \cite{Brown2003b} and \textsc{Mathematica} \cite{Strzebonski10}.
One reason for this improvement is the ability of the new algorithm to recycle subresultant calculations, an idea introduced and detailed in~\cite{CM12a} for the purpose of decomposing polynomial systems into regular chains incrementally.

Another benefit of the incremental approach is that it allows for simplification when constructing a CAD in the presence of ECs.  Instead of working with polynomials, the algorithm can be modified to work with relations.  Then branches in which an EC is not satisfied may be truncated, offering the possibility of a reduction in both computation time and output size.  
In \cite{CM12b} it was shown that using this optimization allowed the algorithm to process examples which \textsc{Mathematica} and \textsc{Qepcad} could not.

\subsection{Contribution and outline}
\label{SUBSEC:Plan}

In Section \ref{SEC:Algorithm} we present our new algorithms.   
Our aim is to combine the savings from an invariance criteria closer to the underlying application, with the savings offered by the case distinction of the regular chains approach.  It requires adapting the existing algorithms for the regular chains approach so that they refine branches of the tree data structure only when necessary for truth-table invariance, and so that branches are truncated only when appropriate to do so.

We implemented our work in the \texttt{RegularChains} library for \textsc{Maple}.  In Section \ref{SEC:Comparison} we qualitatively compare our new algorithm to our previous work and in Section \ref{SEC:Experiments} we present experimental results comparing it to the state of the art.
Finally, in Section \ref{SEC:Conclusion} we give our conclusions and ideas for future work.

\section{Algorithm}
\label{SEC:Algorithm}

\subsection{Constructing a complex cylindrical tree}
\label{SUBSEC:CAlg}

Let $\bm{x}=x_1\prec\cdots\prec x_n$ be a sequence of ordered variables.  We will construct TTICADs of $\R^n$ for a semi-algebraic system $sas$ (Definition \ref{def:sas}).  However, to achieve this we first build CCDs of $\C^n$ with respect to a complex system.
\begin{definition}
\label{def:cs}
Let $F=\{p_1,\ldots,p_s\}$ be a finite set from $\Q[\bm{x}]$, $G\subseteq F$ and
$\sigma_i \in \{=, \neq \}$.  
Then we define a \emph{complex system} (denoted by $cs$) as a set 
\[
\{p_i~\sigma_i~0\mid p_i \in G\}\cup \{p_i\mid p_i \in F\setminus G\}.
\]
\end{definition}
The complex systems we work with will be defined in accordance with a semi-algebraic counterpart (see Definition \ref{def:sas}).
For $p \in \Q[\bm{x}]$ we denote the zero set of $p$ in $\C^n$ by $Z_{\C}(p)$, or $Z_{\C}(p=0)$, and its complement by $Z_{\C}(p\neq 0)$.

We compute CCDs as trees, following \cite{CMXY09,CM12b}.  
Throughout let $T$ be a rooted tree with each node of depth $i$ a polynomial constraint of type either, ``any $x_i$'' (with zero set defined as $\C^n$), or $p=0$, or $p\neq0$ (where $p\in\Q[x_1,\ldots,x_i]$).
For any $i$ denote the induced subtree of $T$ with depth $i$ by $T_i$.  
Let $\Gamma$ be a path of $T$ and define its zero set $Z_{\C}(\Gamma)$ as the intersection of zero sets of its nodes. The zero set of $T$, denoted $Z_{\C}(T)$, is defined as the union of zero sets of its paths. 

\begin{definition}
\label{def:CCT}
$T$ is a {\em complete complex cylindrical tree} (complete CCT) of $\Q[\bm{x}]$ if it satisfies recursively: 
\begin{enumerateshort}
\item If $n=1$: either $T$ has one leaf ``any $x_1$'', 
or it has $s+1$ ($s\geq 1$) leaves $p_1=0,\ldots,p_s=0, \prod_{i=1}^sp_i\neq 0$, 
where $p_i\in\Q[x_1]$ are squarefree and coprime.
\item The induced subtree $T_{n-1}$ is a complete CCT. 
\item For any given path $\Gamma$ of $T_{n-1}$, either its leaf $V$ has only one child ``any $x_n$'', 
or $V$ has $s+1$ ($s\geq 1$) children $p_1=0,\ldots,p_s=0, \prod_{i=1}^sp_i\neq 0$, 
where $p_1,\ldots,p_s\in\Q[\x]$ are squarefree and coprime satisfying: \label{listnum}
\item[\ref{listnum}a.]
for any $\alpha\in Z_{\C}(\Gamma)$, none of ${\rm lc}(p_{j}, x_n)$, $j=1,\ldots,s$, vanishes at $\alpha$, and
\item[\ref{listnum}b.]$p_{1}(\alpha,x_n), \ldots, p_{s}(\alpha,x_n)$ are squarefree and coprime.
\end{enumerateshort}
\end{definition}
The set $\{Z_{\C}(\Gamma)\mid \Gamma~\mbox{is a path of}~ T\}$ is called the {\em complex cylindrical decomposition} (CCD) of $\C^n$ associated with $T$: condition (\ref{listnum}b) assures that it is a decomposition.
Note that for a complete CCT we have $Z_{\C}(T)=\C^n$. A proper subtree rooted at the root node of $T$ of depth $n$ is called a {\em partial CCT} of $\Q[\bm{x}]$. We use CCT to refer to either a complete or partial CCT.
We call a complex cylindrical tree $T$ an \emph{initial tree} if $T$ has only one path and $T$ is complete.

\begin{definition}
Let $T$ be a CCT of $\C^n$ and $\Gamma$ a path of $T$.  
A polynomial $p \in \Q[\bm{x}]$ is \emph{sign invariant} on $\Gamma$ if either 
$Z_{\C}(\Gamma) \cap Z_{\C}(p)=\emptyset$ 
or 
$Z_{\C}(p) \supseteq Z_{\C}(\Gamma)$.
A constraint $p=0$ or $p\neq 0$ is \emph{truth-invariant} on $\Gamma$ if $p$ is sign-invariant on $\Gamma$.
A complex system $cs$ is \emph{truth-invariant} on $\Gamma$ if the conjunction of the constraints in $cs$ is truth-invariant on $\Gamma$, and each polynomial in $cs$ is sign-invariant on $\Gamma$.
\end{definition}

\begin{example}
Let $q := (x_2^2+x_2+x_1)$ and $p := x_1q$.  The following tree is a CCT such that $p$ is sign-invariant (and $p=0$ is truth invariant) on each path.
\begin{center}
\begin{small}
\qtreecenterfalse
\Tree [.r [.$4x_1-1=0$ $2x_2+1=0$ $2x_2+1\neq0$ ] [.$x_1=0$ {{\rm any} $x_2$} ] !\qsetw{1cm} [.$x_1(4x_1-1)\neq 0$ $q=0$ $q\neq0$ ]  ]
\end{small}
\end{center}
\end{example}

\noindent We introduce Algorithm \ref{Algo:CCD} to produce truth-table invariant CCTs, and new sub-algorithms \ref{Algo:IntersectLCS} and \ref{Algo:IntersectPolySet}.  It also uses \texttt{IntersectPath} and \texttt{NextPathToDo} from \cite{CM12b}.
\texttt{IntersectPath} takes: a CCT $T$; a path $\Gamma$; and $p$, either a polynomial or constraint.  When a polynomial it refines $T$ so $p$ is sign-invariant above each path from $\Gamma$ (still satisfying Definition \ref{def:CCT}). When a constraint it refines so the constraint is true, possibly truncating branches if there can be no solution.  This necessitates the housekeeping algorithm \texttt{MakeComplete} which restores to a complete CCT by simply adding missing siblings (if any) to every node.
\texttt{NextPathToDo} simply
returns the next incomplete path $\Gamma$ of $T$.  

\begin{prop}
\label{Prop:AlgILCS}
Algorithm~\ref{Algo:CCD} satisfies its specification.
\end{prop}
\begin{proof}
It suffices to show that Algorithm~\ref{Algo:IntersectLCS} is as specified.  First observe that Algorithm \ref{Algo:IntersectPolySet} just recursively calls \texttt{IntersectPath} on constraints and so its correctness follows from that of \texttt{IntersectPath}.  When called on ECs \texttt{IntersectPath} may return a partial tree and so \texttt{MakeComplete} must be used in line \ref{branchMC}.

Algorithm \ref{Algo:IntersectLCS} is clearly correct is its base cases, namely 
line~\ref{branch1}, line~\ref{branch2} and line~\ref{branch3}. 
It also clearly terminates since the input of each recursive call has less constraints. 
For each path $C$ of the refined $\Gamma$, by induction, it is sufficient to show that $cs$ is truth-invariant on $C$.   If $p\neq 0$ on $C$, then $cs$ is false on $C$.   If $p=0$ on $C$, then the truth of $cs$ is invariant since it is completely determined by the truth of $cs' := cs \setminus \{p=0\}$, invariant on $C$ by induction. 
\end{proof}

\begin{algorithm}[h]
\label{Algo:CCD}
\caption{${\sf TTICCD(L)}$}
\KwIn{
A list $L$ of complex systems of $\Q[\bm{x}]$. 
}
\KwOut{
A complete CCT $T$ with each $cs \in L$ truth-invariant on each path.
}
Create the initial CCT $T$ and let $\Gamma$ be its path\;
${\sf IntersectLCS}(L, \Gamma, T)$\;
\end{algorithm}

\begin{algorithm}[H]
\label{Algo:IntersectLCS}
\caption{${\sf IntersectLCS}(L, \Gamma, T)$}
\KwIn{
A CCT $T$ of $\Q[\bm{x}]$.
A path $\Gamma$ of $T$. 
A list of complex systems $L$.
}
\KwOut{
Refinements of $\Gamma$ and $T$ such that $T$ is complete, and $cs \in L$ is truth-invariant above each path of $\Gamma$.
}
\uIf{$L=\emptyset$}{
\texttt{return}\label{branch1}\;
}
\uElseIf{$|L|=1$}{
    Let $cs$ be the only complex system\;
    ${\sf IntersectPolySet}(cs, \Gamma, T)$\;\label{branch2}
    ${\sf MakeComplete}(T)$\; \label{branchMC}
}
\uElseIf{no $cs \in L$ has an equational constraint}{
    Let $F$ be the set of polynomials appearing in $L$\;
    ${\sf IntersectPolySet}(F, \Gamma, T)$\label{branch3}\;
}
\Else{
   Let $cs$ be a complex system of $L$ with an EC denoted $p=0$\label{step:choice}\;
   \label{branch4}  
   ${\sf IntersectPath}(p, \Gamma, T)$ \tcp{$\Gamma$ may become a tree}
   \While{$C := {\sf NextPathToDo}(\Gamma)\neq \emptyset$}{
          \uIf{$p=0$ on $C$}{
              $cs' := cs \setminus \{p=0\}$\;
                 ${\sf IntersectLCS}(L\setminus\{cs\}\cup \{cs'\}, C, T)$\;
          }
          \Else{
              ${\sf IntersectLCS}(L\setminus\{cs\}, C, T)$\;
          }
    }
}
\end{algorithm}
\,
\begin{algorithm}[H]
\label{Algo:IntersectPolySet}
\caption{${\sf IntersectPolySet}(F, \Gamma, T)$}
\KwIn{
A CCT $T$, a path $\Gamma$ and 
a set $F$ of polynomials (constraints).
}
\KwOut{
$T$ is refined and $\Gamma$ becomes a subtree.  
Each polynomial (constraint) of $F$ is sign (truth)-invariant above each path of $\Gamma$.
}
\lIf{$F=\emptyset$}{return\;}
Let $p\in F$; $F' := F \setminus \{p\}$\;
${\sf IntersectPath}(p, \Gamma, T)$; \tcp{$\Gamma$ may become a tree}
\If{$F'\neq \emptyset$}{
\While{$C := {\sf NextPathToDo}(\Gamma)\neq \emptyset$}{
         ${\sf IntersectPolySet}(F', C, T)$\;
    }
}
\end{algorithm}

\subsection{Illustrating the computational flow}
\label{SUBSEC:Flow}

Consider using Algorithm \ref{Algo:CCD} on input of the form \\
$\hspace*{0.8in} L = [cs_1, cs_2] := [\{f_1=0, g_1 \neq 0\}, \, \{f_2=0, g_2\neq 0\}].$\\
Algorithm \ref{Algo:CCD} constructs the initial tree and passes to Algorithm \ref{Algo:IntersectLCS}.  We enter the fourth branch of the conditional, let $p=f_1$, and refine to a sign invariant CCT for $f_1$.  This makes a case distinction between $f_1=0$ and $f_1\neq 0$. 
On the branch $f_1\neq 0$, we recursively call ${\sf IntersectLCS}$ on $[cs_2]$ which then passes directly to ${\sf IntersectPolySet}$.  
On the branch $f_1=0$, we recursively call ${\sf IntersectLCS}$ on $[\{g_1\neq 0\}, \{f_2=0, g_2\neq 0\}]$.
This time $p=f_2$ and a case discussion is made between $f_2=0$ and $f_2\neq 0$. 
On the branch $f_2 \neq 0$, we end up calling ${\sf IntersectPolySet}(g_1\neq 0)$ while on the branch $f_2=0$ we call ${\sf IntersectLCS}$ on $[\{g_1\neq 0\}, \{g_2\neq 0\}]$, which reduces to ${\sf IntersectPolySet}(g_1, g_2)$. 
The case discussion is summarised by:
\[
\left\{
\begin{array}{ll}
f_1=0: &
\left\{ 
\begin{array}{ll}
f_2=0: & g_1, g_2\\
f_2\neq 0: & g_1\neq 0
\end{array}
\right. \\
f_1\neq 0: & f_2=0, g_2\neq 0
\end{array}
\right..
\]

\subsection{Refining to a TTICAD}
\label{SUBSEC:RAlg}

We now discuss how Section \ref{SUBSEC:CAlg} can be extended from CCDs to CADs.  
\begin{definition}
\label{def:sas}
A \emph{semi-algebraic system} of $\Q[\bm{x}]$ ($sas$) is a set of constraints $\{p_i~\sigma_i~0\}$ where each $\sigma_i \in \{=, >, \geq,\neq \}$ and each $p_i \in \Q[\bm{x}]$.
A \emph{corresponding complex system} is formed 
by replacing all $p_i>0$ by $p_i \neq 0$ and all $p_i\geq 0$ by $p_i$.

A $sas$ is \emph{truth-invariant} on a cell if the conjunction of its constraints is.
\end{definition}
Note that the ECs of an sas are still identified as ECs of the corresponding cs.  Algorithm \ref{Algo:RCTTI} produces a TTICAD of $\R^n$ for a sequence of semi-algebraic systems.  \\

\begin{algorithm}[H]
\label{Algo:RCTTI}
\caption{${\sf RC-TTICAD(L)}$}
\KwIn{
A list $L$ of semi-algebraic systems of $\Q[\bm{x}]$. 
}
\KwOut{
A CAD such that each $sas \in L$ is truth-invariant on each cell.
}
Set $L'$ to be the list of corresponding complex systems\;
$\mathcal{D} := {\sf TTICCD}(L')$\;
\texttt{MakeSemiAlgebraic}($\mathcal{D}, n$)\;
\end{algorithm}

\begin{prop}
Algorithm \ref{Algo:RCTTI} satisfies its specification.
\end{prop}
\begin{proof}
Algorithm \ref{Algo:RCTTI} starts by building the corresponding $cs$ for each $sas$ in the input.  It uses Algorithm \ref{Algo:CCD} to form a CCD truth-invariant for each of these and then the algorithm \texttt{MakeSemiAlgebraic} introduced in \cite{CMXY09} to move to a CAD.  
\texttt{MakeSemiAlgebraic} takes a CCD $\mathcal{D}$ and outputs a CAD $\mathcal{E}$ such that for each element $d \in \mathcal{D}$ the set $d \cap \R^n$ is a union of cells in $\mathcal{E}$.  Hence $\mathcal{E}$ is still truth-invariant for each $cs \in L'$.  It is also a TTICAD for $L$, (as to change sign from positive to negative would mean going through zero and thus changing cell).  
The correctness of Algorithm \ref{Algo:RCTTI} hence follows from the correctness of its sub-algorithms.
\end{proof}

The output of Algorithm \ref{Algo:RCTTI} is a TTICAD for the formula defined by each semi-algebraic system (the conjunction of the individual constraints of that system).  To consider formulae with disjunctions we must first use disjunctive normal form and then construct semi-algebraic systems for each conjunctive clause.

\section{Comparison with prior work}
\label{SEC:Comparison}

\noindent We now compare qualitatively to our previous work.  Quantitative experiments and a comparison with competing CAD implementations follows in Section \ref{SEC:Experiments}.

\subsection{Comparing with sign-invariant CAD by regular chains}
\label{SUBSEC:VsRC-CAD}

Algorithm \ref{Algo:RCTTI} uses work from \cite{CM12b} but obtains savings when building the complex tree by ensuring only truth-table invariance.  To demonstrate this we compare diagrams representing the number of times a constraint is considered when building a CCD for a complex system. 

\begin{definition}
Let $cs$ be a complex system.
We define the {\em complete (resp. partial) combination diagram} for $cs$, denoted by $\Delta_0(cs)$ (resp. $\Delta_1(cs)$), recursively: 
\begin{itemizeshort}
\item If $cs=\emptyset$, then $\Delta_i(cs)$ ($i=0,1$) is defined to be null. 
\item If $cs$ has any ECs then select one, $\psi$ (defined by a polynomial $f$), and define
\begin{align*}
\Delta_0(cs) &:= \left\{
\begin{array}{lr}
f=0     & \Delta_0(cs\setminus\{\psi\}) \\
f\neq 0 & \Delta_0(cs\setminus\{\psi\})
\end{array}
\right., \qquad
\Delta_1(cs) := \left\{
\begin{array}{lr}
f=0     & \Delta_1(cs\setminus\{\psi\}) \\
f\neq 0 & 
\end{array}
\right..
\end{align*}
\item Otherwise select a constraint $\psi$ (which is either of the form $f\neq 0$, or $f$) and for $i=0,1$ define \\
$ \hspace*{1.3in}
\Delta_i(cs) := \left\{
\begin{array}{lr}
f=0 & \Delta_i(cs\setminus\{\psi\})\\
f\neq 0 & \Delta_i(cs\setminus\{\psi\})\\
\end{array}
\right..
$
\end{itemizeshort}
\end{definition}
The combination diagrams illustrate the combinations of relations that must be analysed by our algorithms, with the partial diagram relating to Algorithm \ref{Algo:CCD} and the complete diagram the sign-invariant algorithm from \cite{CM12b}.

\begin{lem}
\label{Lemma:Counting}
Assume that the complex system $cs$ has $s$ ECs and $t$ constraints of other types. 
Then the number of constraints appearing in $\Delta_0(cs)$ is $2^{s+t+1}-2$, and the number appearing in $\Delta_1(cs)$ is $2(2^{t}+s)-2$. 
\end{lem}
\begin{proof}
The diagram $\Delta_0(cs)$ is a full binary tree with depth $s+t$.  Hence the number of constraints appearing is the geometric series $\sum_{i=1}^{s+t} 2^i = 2^{s+t+1}-2$.  

$\Delta_1(cs)$ will start with a binary tree for the ECs, with only one branch continuing at each depth, and thus involves $2s$ constraints.  The full binary tree for the other constraints is added to the final branch, giving a total of $2^{t+1}+2s-2$.
\end{proof}

\begin{definition}
Let $L$ be a list of complex systems.  
We define the {\em complete (resp. partial) combination diagram} of $L$, 
denoted by $\Delta_0(L)$ (resp. $\Delta_1(L)$) recursively:
If $L=\emptyset$, then $\Delta_i(L)$, $i=0,1$, is null. 
Otherwise let $cs$ be the first element of $L$.
Then $\Delta_i(L)$ is obtained by appending $\Delta_i(L \setminus \{cs\})$ to each leaf node of $\Delta_i(cs)$.
\end{definition}

\begin{theorem}
Let $L$ be a list of $r$ complex systems. 
Assume each $cs \in L$ has $s$ ECs and $t$ constraints of other types.
Then the number of constraints appearing in $\Delta_0(L)$ is 
$2^{r(s+t)+1}-2$ and the number of constraints appearing in $\Delta_1(L)$ is $N(r)=2(s+2^t)^r-2$.
\end{theorem}
\begin{proof}
The number of constraints in $\Delta_0(L)$ again follows from the geometric series.  
For $\Delta_1(L)$ we proceed with induction on $r$. The case $r=1$ is given by Lemma \ref{Lemma:Counting}, so now assume $N(r-1)=2(s+2^t)^{r-1}-2$.  

The result for $r$ follows from $C(r) = C(1) + (s+2^t)C(r-1)$. To conclude this identity consider the diagram for the first $cs \in L$.  To extend to $\Delta_1(L)$ we append $\Delta_1(L \setminus {cs})$ to each end node.  There are $s$ for cases where an EC was not satisfied and $2^t$ from cases where all ECs were (and non-ECs were included).
\end{proof}

\begin{figure}[t]
\caption{The left is a sign-invariant CAD, and the right a TTICAD, for (\ref{eq:Form1}).}
\label{fig:Ex1}
\centering
\includegraphics[width=0.48\textwidth]{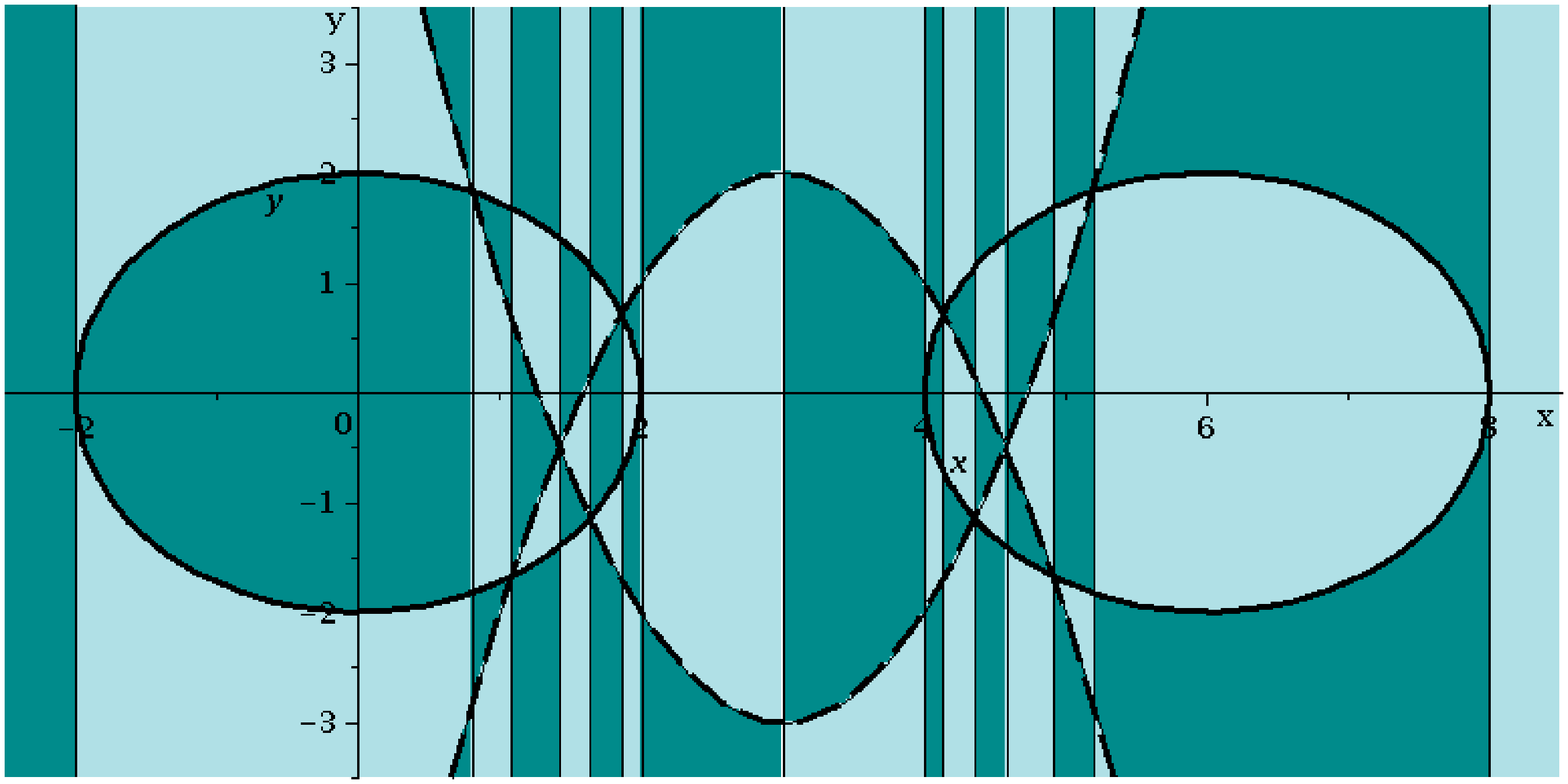}
\hspace*{0.02\textwidth}
\includegraphics[width=0.48\textwidth]{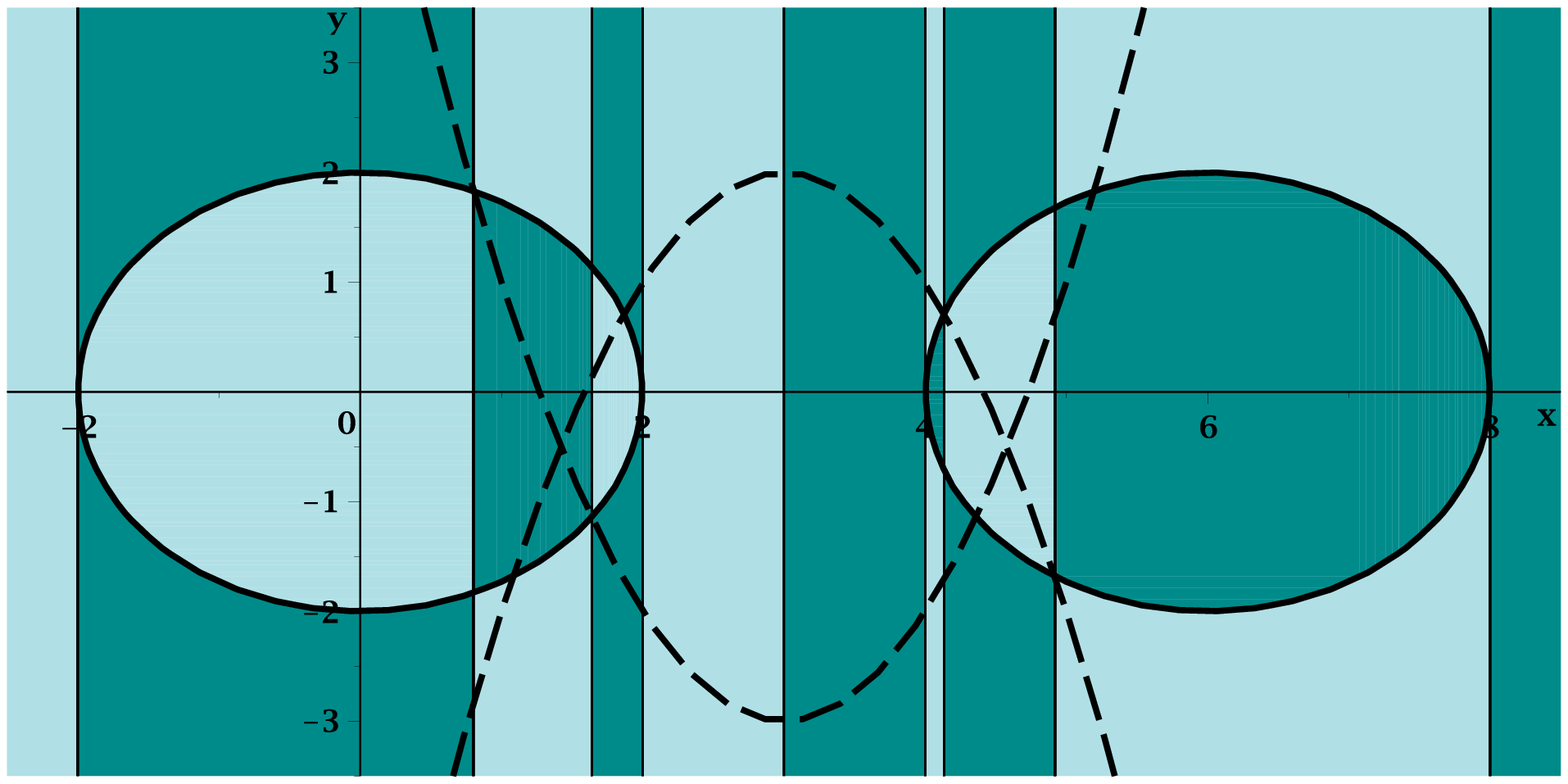}
\end{figure}

\begin{example}
\label{ex:1a}
\noindent We demonstrate these savings by considering
\begin{align}
f_1 &:= x^2+y^2-4,     \hspace*{0.44in} 
g_1 := (x-3)^2-(y+3),  \hspace*{0.1in} 
\phi_1 := f_1 = 0 \land g_1 < 0, \nonumber \\
f_2 &:= (x-6)^2+y^2-4, \,\,\,\,
g_2 := (x-3)^2+(y-2),  \,\,\,\, 
\phi_2 := f_2 = 0 \land g_2 < 0,
\label{eq:Form1}
\end{align}
and ordering $x \prec y$.  
The polynomials are graphed in Figure \ref{fig:Ex1} where the solid circles are the $f_i$ and the dashed parabola the $g_i$.
To study the truth of the formulae $\{\phi_1,\phi_2\}$ we could create a sign-invariant CAD.  Both the incremental regular chains technology of \cite{CM12b} and {\sc Qepcad} \cite{Brown2003b} do this with 231 cells.  The 72 full dimensional cells are visualised on the left of Figure \ref{fig:Ex1}, (with the cylinders on each end actually split into three full dimensional cells out of the view).

Alternatively we may build a TTICAD using Algorithm \ref{Algo:RCTTI} to obtain only 63 cells, 22 of which have full dimension as visualised on the right of Figure \ref{fig:Ex1}. By comparing the figures we see that the differences begin in the CAD of the real line, with the sign-invariant case splitting into 31 cells compared to 19.
The points identified on the real line each align with a feature of the polynomials.  Note that the TTICAD identifies the intersections of $f_i$ and $g_j$ only when $i=j$, and that no features of the inequalities are identified away from the ECs.
\end{example}

\subsection{Comparing with TTICAD by projection and lifting}
\label{SUBSEC:RCvsPL}

We now compare Algorithm \ref{Algo:RCTTI} with the TTICADs obtained by projection and lifting in \cite{BDEMW13}.  We identify three main benefits which we demonstrate by example. \\

\noindent \textbf{(I) Algorithm \ref{Algo:RCTTI} can achieve cell savings from case distinction.}
\begin{example}
\label{ex:1b}
Algorithm \ref{Algo:RCTTI} produces a TTICAD for (\ref{eq:Form1}) with 63 cells
compared to a TTICAD of 67 cells from the projection and lifting algorithm in \cite{BDEMW13}.  The full-dimensional cells are identical and so the image on the right of Figure \ref{fig:Ex1} is an accurate visualisation of both.  
To see the difference we must compare lower dimensional cells.  Figure \ref{fig:Ex1Zoom} compares the lifting to $\R^2$ over a cell on the real line aligned with an intersection of $f_1$ and $g_1$.    
The left concerns the algorithm in \cite{BDEMW13} and the right Algorithm \ref{Algo:RCTTI}.  The former isolates both the $y$-coordinates where $f_1=0$ while the latter only one (the single point over the cell where $\phi_1$ is true).  
\end{example}

\begin{figure}[t]
\caption{Comparing TTICADs for (\ref{eq:Form1}).  The left uses \cite{BDEMW13} and the right Algorithm \ref{Algo:RCTTI}.}
\label{fig:Ex1Zoom}
\centering
\includegraphics[width=0.27\textwidth]{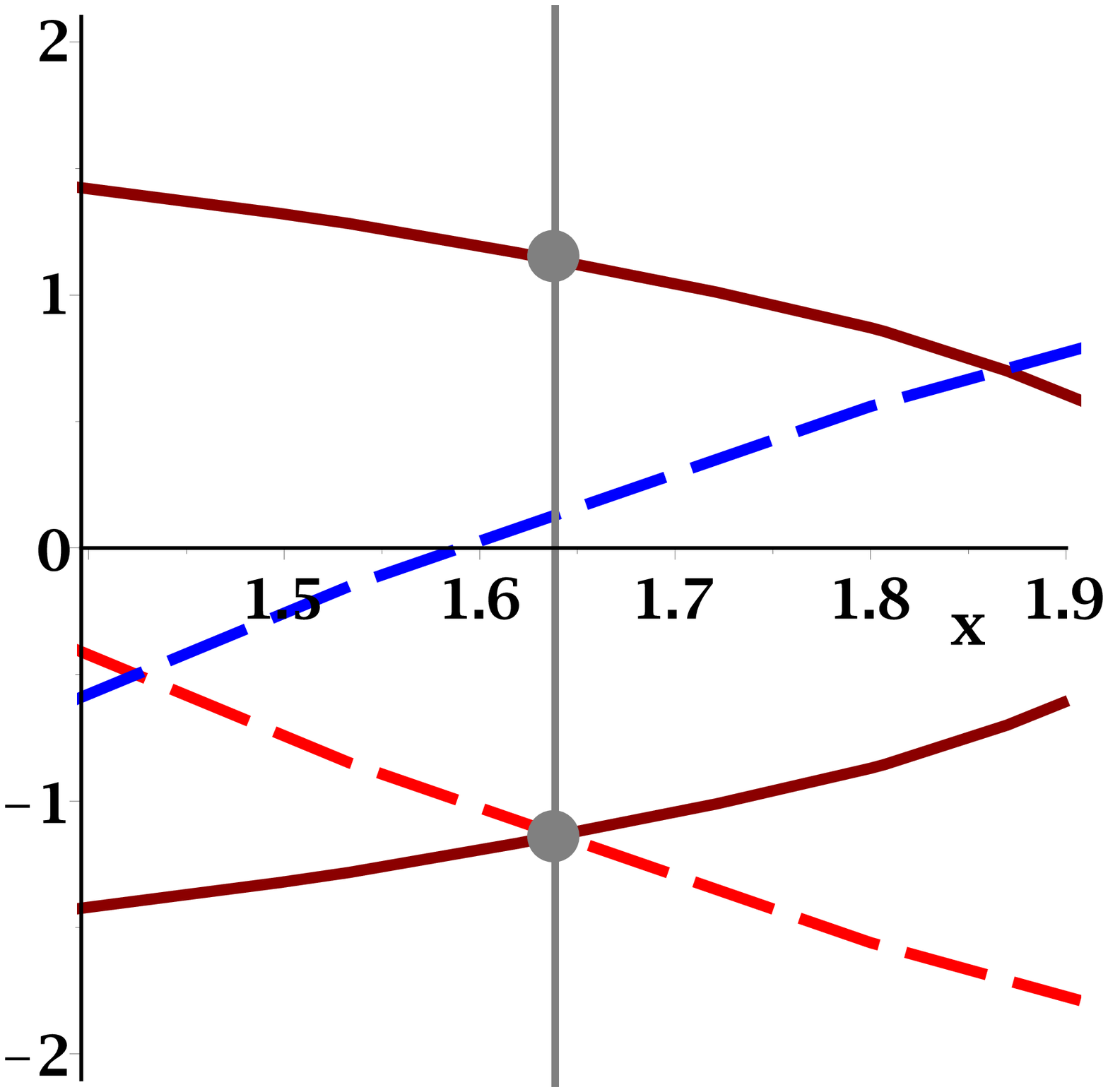}
\hspace*{0.2\textwidth}
\includegraphics[width=0.27\textwidth]{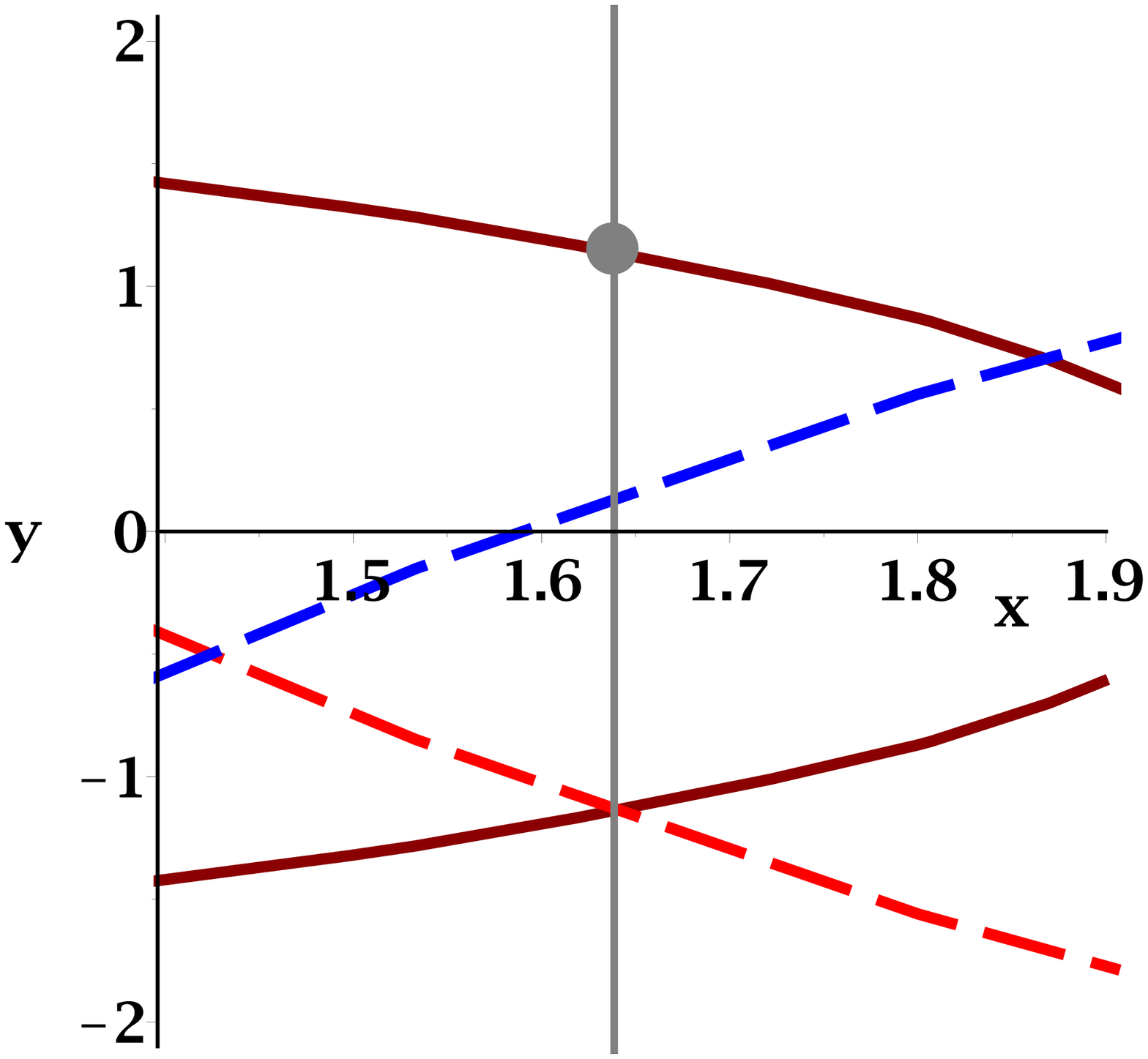}
\end{figure}

If we modified the problem so the inequalities in (\ref{eq:Form1}) were not strict then $\phi_1$ becomes true at both points and Algorithm \ref{Algo:RCTTI} outputs the same TTICAD as \cite{BDEMW13}.  Unlike \cite{BDEMW13}, the type of the non-ECs affects the behaviour of Algorithm \ref{Algo:RCTTI}.  

\vspace*{5pt}

\noindent \textbf{(II) Algorithm \ref{Algo:RCTTI} can take advantage of more than one EC per clause.}
\begin{example}
\label{ex:bieq}
We assume $x \prec y$ and consider
\begin{align}
f_1 &:= x^2+y^2-1, \qquad h := y^2-\tfrac{x}{2}, \qquad g_1 := xy - \tfrac{1}{4}    
\nonumber  \\
f_2 &:= (x-4)^2+(y-1)^2-1  \quad g_2 := (x-4)(y-1) - \tfrac{1}{4}, 
\nonumber \\
\phi_1 &:= h=0 \land f_1 = 0 \land g_1 < 0, \,\, \phi_2 := f_2 = 0 \land g_2 < 0.
\label{eq:Form2}
\end{align}
The polynomials are graphed in Figure \ref{fig:Ex2} where the dashed curves are $f_1$ and $h$, the solid curve is $f_2$ and the dotted curves are $g_1$ and $g_2$.  A TTICAD produced by Algorithm \ref{Algo:RCTTI} has 69 cells and is visualised on the right of Figure \ref{fig:Ex2} while a TTICAD produced by projection and lifting has 117 cells and is visualised on the left.  This time the differences are manifested in the full-dimensional cells.  

The algorithm from \cite{BDEMW13} works with a single designated EC in each QFF (in this case we chose $f_1$) and so treats $h$ in the same way as $g_1$.  This means for example that all the intersections of $h$ or $g_1$ with $f_1$ are identified.  By comparison, Algorithm \ref{Algo:RCTTI} would only identify the intersection of $g_1$ with an EC if this occurred at a point where both $f_1$ and $h$ were satisfied (does not occur here). 
For comparison, a sign-invariant CAD using {\sc Qepcad} or \cite{CM12b} has 611 cells.  
\end{example}

To use \cite{BDEMW13} we had to designate either $f_1$ or $h$ as the EC.  Choosing $f_1$ gave 117 cells and $h$ 163.  Our new algorithm has similar choices: what order should the systems be considered and what order the ECs within (step \ref{step:choice} of Algorithm \ref{Algo:IntersectLCS})?  Processing $f_1$ first gives 69 cells but other choice can decrease this to 65 or increase it to 145.  See \cite{EBCDMW14} for advice on making such choices intelligently.

\begin{figure}[t]
\caption{TTICAD for (\ref{eq:Form2}).  The left uses \cite{BDEMW13} and the right Algorithm \ref{Algo:RCTTI}.}
\label{fig:Ex2}
\centering
\includegraphics[width=0.48\textwidth]{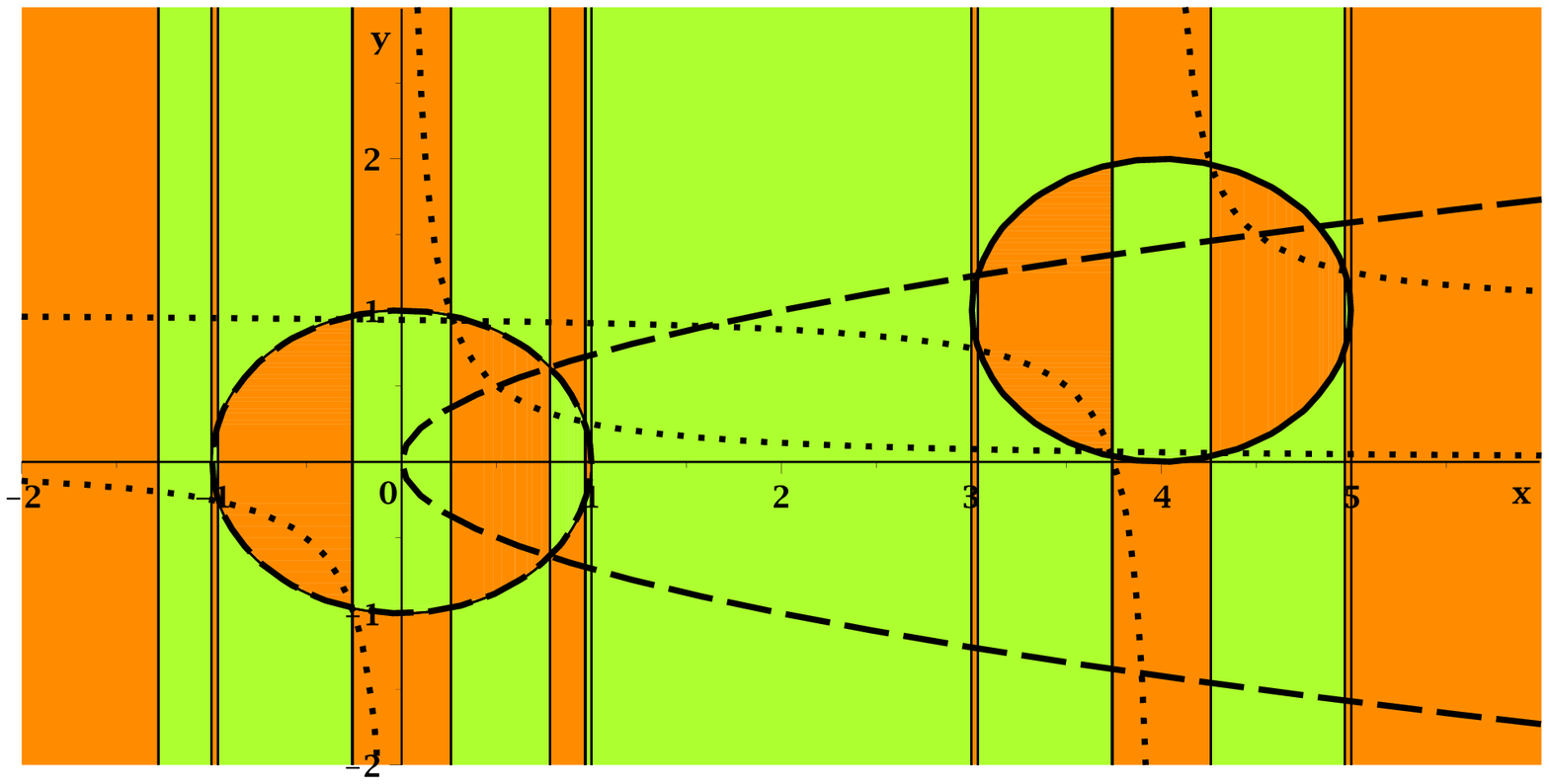}
\hspace*{0.02\textwidth}
\includegraphics[width=0.48\textwidth]{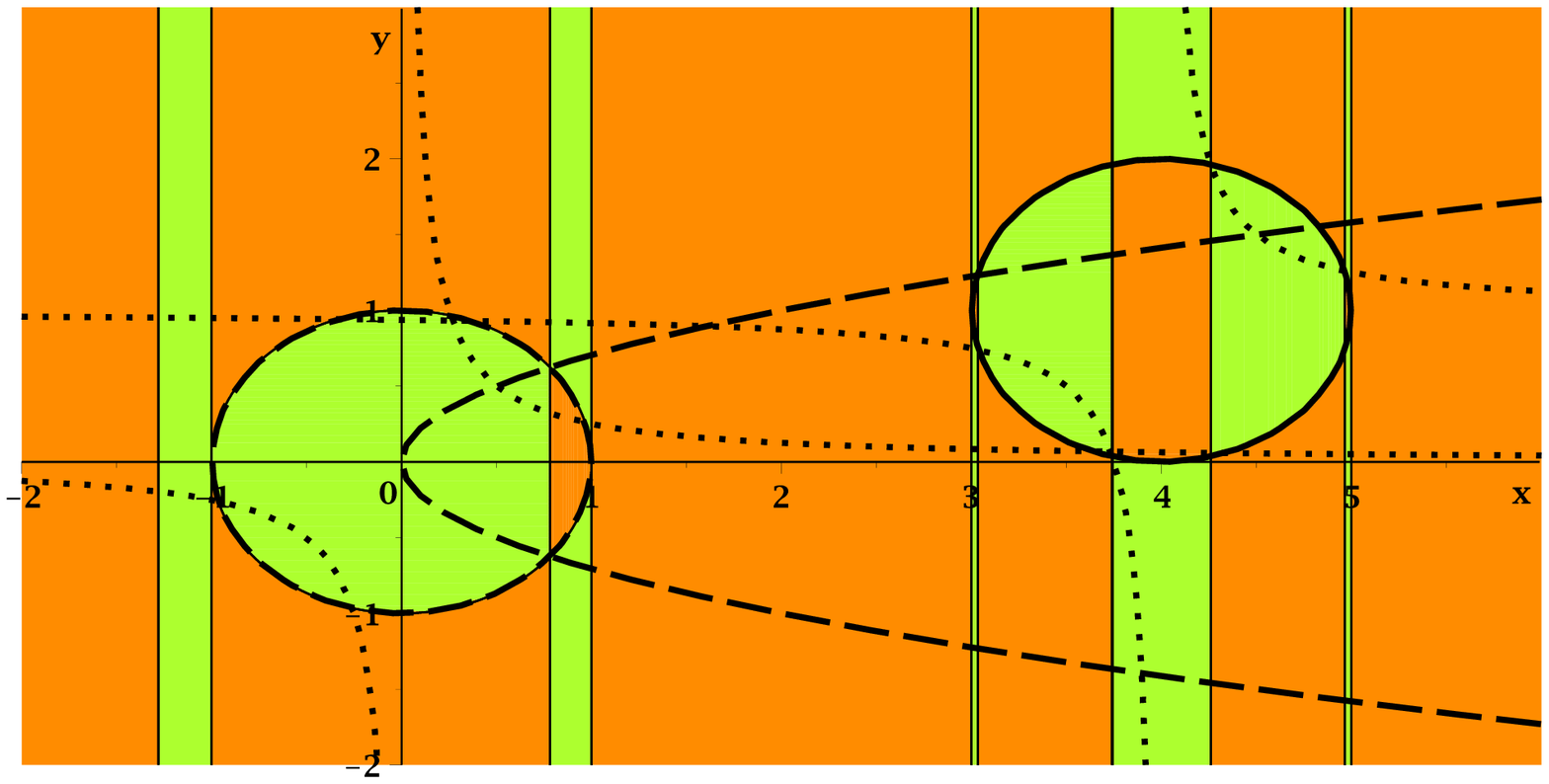}
\end{figure}

\vspace*{5pt}

\noindent \textbf{(III) Algorithm \ref{Algo:RCTTI} will succeed given sufficient time and memory.}

This contrasts with the theory of reduced projection operators used in \cite{BDEMW13}, where input must be \emph{well-oriented} (meaning that certain projection polynomials cannot be nullified when lifting over a cell with respect to them).
\begin{example}
\label{ex:bcwo}
Consider the identity $\sqrt{z}\sqrt{w} = \sqrt{zw}$ over $\mathbb{C}^2$.  We analyse its truth by decomposing according to the branch cuts and testing each cell at its sample point.  
Letting $z=x+\rm{i}y, w = u + \rm{i}v$ we see that branch cuts occur when 
$
\hspace*{0.5in} (y = 0 \land x < 0) \lor (v = 0 \land u < 0) \lor (yu+xv = 0 \land xu-yv < 0).
$\\
We desire a TTICAD for the three clauses joined by disjunction.  Assuming $v \prec u \prec x \prec y$ Algorithm \ref{Algo:RCTTI} does this using 97 cells, while the projection and lifting approach identifies the input as not well-oriented.  The failure is triggered by $yu+xv$ being nullified over a cell where $u=x=0$ and $v<0$.  
\end{example}

\section{Experimental Results}
\label{SEC:Experiments}

We present experimental results obtained on a Linux desktop (3.1GHz Intel processor, 8.0Gb total memory).  
We tested on 52 different examples, with a representative subset of these detailed in Table \ref{tab:Results}. 
The examples and other supplementary information are available from \url{http://opus.bath.ac.uk/38344/}.
One set of problems was taken from CAD papers \cite{BH91} \cite{BDEMW13} and a second from system solving papers \cite{CGLMP07} \cite{CM12b}.  The polynomials from the problems were placed into different logical formulations: disjunctions in which every clause had an EC (indicated by $\dagger$) and disjunctions in which only some do (indicated by $\dagger\dagger$).  A third set was generated by branch cuts of algebraic relations: addition formulae for elementary functions and examples previously studied in the literature.  

Each problem had a declared variable ordering
(with $n$ the number of variables).  For each experiment a CAD was
produced with the time taken (in seconds) and number of cells (cell
count) given. The first is an obvious metric and the second crucial
for applications acting on each cell.
T/O indicates a time out (set at $30$ minutes), FAIL a failure due to theoretical reasons such as input not being well-oriented (see \cite{McCallum1999} \cite{BDEMW13}) and Err an unexpected error.

We start by comparing with our previous work (all tested in \textsc{Maple} 18) by considering the first five algorithms in Table \ref{tab:Results}.  
RC-TTICAD is Algorithm \ref{Algo:RCTTI}, PL-TTICAD the algorithm from \cite{BDEMW13}, PL-CAD CAD with McCallum projection, RC-Inc-CAD the algorithm from \cite{CM12b} and RC-Rec-CAD the algorithm from \cite{CMXY09}.  Those starting RC are part of the \texttt{RegularChains} library and those starting PL the \texttt{ProjectionCAD} package \cite{EWBD14}.  RC-Rec-CAD is a modification of the algorithm currently distributed with \textsc{Maple}; the construction of the CCD is the same but the conversion to a CAD has been improved.  Algorithms RC-TTICAD and RC-Rec-CAD are currently being integrated into the {\tt RegularChains} library, which can be downloaded from \url{www.regularchains.org}.

We see that RC-TTICAD never gives higher cell counts than any of our previous work and that in general the TTICAD theories allow for cell counts an order of magnitude lower.  
RC-TTICAD is usually the quickest in some cases offering vast speed-ups.  
It is also important to note that there are many examples where PL-TTICAD has a theoretical failure but for which RC-TTICAD will complete (see point (III) in Section \ref{SUBSEC:RCvsPL}).  Further, these failures largely occurred in the examples from branch cut analysis, a key application of TTICAD.

We can conclude that our new algorithm combines the good features of our previous approaches, giving an approach superior to either.  We now compare with competing CAD implementations, detailed in the last four columns of Table \ref{tab:Results}: {\sc Mathematica} \cite{Strzebonski10} (V9 graphical interface); {\sc Qepcad-B} \cite{Brown2003b} (v1.69 with options {\tt +N500000000} {\tt +L200000}, initialization included in the timings and implicit EC declared when present); the \textsc{Reduce} package \textsc{Redlog} \cite{DSS04} (2010 Free CSL Version); and the \textsc{Maple} package \textsc{SyNRAC} (2011 version) \cite{IYAY09}.

As reported in \cite{BDEMW13}, the TTICAD theory allows for lower cell
counts than \textsc{Qepcad} even when manually declaring an EC.  We found that both \textsc{SyNRAC} and \textsc{Redlog} failed for many examples, (with \textsc{SyNRAC} returning unexpected errors and \textsc{Redlog} hanging with no output or messages).  There were examples for which
\textsc{Redlog} had a lower cell count than RC-TTICAD due to the use
of partial lifting techniques, but this was not the case in general.
We note that we were using the most current public version of
\textsc{SyNRAC} which has since been replaced by a superior
development version, (to which we do not have access) and that
\textsc{Redlog} is mostly focused on the virtual substitution approach
to quantifier elimination but that we only tested the CAD command.

{\sc Mathematica} is the quickest in general, often impressively so.
However, the output there is not a CAD but a formula with a
cylindrical structure \cite{Strzebonski10} (hence cell counts are not
available).  Such a formula is sufficient for many applications (such
as quantifier elimination) but not for others (such as algebraic
simplification by branch cut decomposition).  Further, there are 
examples for which RC-TTICAD completes but {\sc Mathematica} times
out.  {\sc Mathematica}'s output only references the CAD cells for
which the input formula is true.  
Our implementation can be modified to do this and in some cases this
can lead to significant time savings; we will investigate this further
in a later publication.

Finally, note that the TTICAD theory allows algorithms to change
with the logical structure of a problem.  For example,
Solotareff$\dagger$ is simpler than Solotareff$\dagger\dagger$ (it has an inequality instead of an equation).  A smaller
TTICAD can hence be produced, while sign-invariant algorithms give
the same output.

\section{Conclusions and further work}
\label{SEC:Conclusion}

We presented a new CAD algorithm which uses truth-table invariance, to give output aligned to underlying problem, and regular chains technology, bringing the benefits of case distinction and no possibility of theoretical failure from well-orientedness conditions.
However, there are still many questions to be considered:
\begin{itemizeshort}
\item Can we make educated choices for the order systems and constraints are analysed by the algorithm?  Example \ref{ex:bieq} and \cite{EBCDMW14} shows this could be beneficial.
\item Can we use heuristics to make choices such as what variable ordering (see current work in \cite{EBDW14} and previous work in \cite{DSS04} \cite{BDEW13}).
\item Can we modify the algorithm for the case of providing truth invariant CADs for a formula in disjunctive normal form?  In this case we could cease refinement in the complex tree once a branch is known to be true.
\item Can we combine with other theory such as partial CAD \cite{CH91} or cylindrical algebraic sub-decompositions \cite{WBDE14}?  
\end{itemizeshort}

\subsubsection*{Acknowledgements}
Supported by the CSTC (grant cstc2013jjys0002), the EPSRC (grant EP/J003247/1) and the NSFC (grant 11301524).

\begin{sidewaystable}[p]
\centerline{
\begin{tabular}{lcccccccccccccccccc}
 &
 & \multicolumn{2}{c}{RC-TTICAD} 
 & \multicolumn{2}{c}{RC-Inc-CAD} 
 & \multicolumn{2}{c}{RC-Rec-CAD} 
 & \multicolumn{2}{c}{PL-TTICAD} 
 & \multicolumn{2}{c}{PL-CAD} 
 & \multicolumn{1}{c}{\textsc{Mathematica}} 
 & \multicolumn{2}{c}{\textsc{Qepcad}} 
 & \multicolumn{2}{c}{\textsc{SyNRAC}} 
 & \multicolumn{2}{c}{\textsc{Redlog}}  
 \\
\cmidrule(lr){3-4}
\cmidrule(lr){5-6}
\cmidrule(lr){7-8}
\cmidrule(lr){9-10}
\cmidrule(lr){11-12}
\cmidrule(lr){13-13}
\cmidrule(lr){14-15}
\cmidrule(lr){16-17}
\cmidrule(lr){18-19}
Problem & n
& Cells & Time 
& Cells & Time 
& Cells & Time 
& Cells & Time 
& Cells & Time 
& Time 
& Cells & Time 
& Cells & Time 
& Cells & Time  
\\
\midrule
Intersection$\dagger$ & 3 
& 541   & 1.0            
& 3723 & 12.0   
& 3723 & 19.0        
& 579 & 3.5 
& 3723 & 29.5
& 0.1  
& 3723  & 4.9   
& 3723  & 12.8 
& Err & --- 
\\
Ellipse$\dagger$ & 5
& 71231  & 317.1         
& 81183 & 544.9 
& 81193 & 786.8 
& FAIL & --- 
& FAIL & --- 
& 11.2 
& 500609 & 275.3 
& Err  & ---    
& Err & --- 
\\
Solotareff$\dagger$ & 4
& 2849   & 8.8           
& 54037  & 209.1   
& 54037 & 539.0 
& FAIL & --- 
& 54037 & 407.6
& 0.1 
& 16603  & 5.2   
& Err  & ---    
& 3353 & 8.6 
\\
Solotareff$\dagger \dagger$ & 4
& 8329    & 21.4              
& 54037 & 226.9 
& 54037 & 573.4
& FAIL & --- 
& 54037 & 414.3
& 0.1  
& 16603  & 5.3   
& Err  & ---    
& 8367 & 13.6 
\\
2D Ex$\dagger$ & 2
& 97       & 0.2                 
& 317 & 1.0   
& 317 & 2.6   
& 105 & 0.6 
& 317 & 1.8
& 0.0 
& 249   & 4.8   
& 317   & 1.1  
& 305 & 0.9 
\\
2D Ex$\dagger \dagger$ & 2
& 183      & 0.4                
& 317 & 1.1
& 317 & 2.6             
& 183 & 1.1 
& 317 & 1.8
& 0.0  
& 317   & 4.6   
& 317   & 1.2  
& 293 & 0.9 
\\
3D Ex$\dagger$ & 3
& 109     & 3.5                 
& 3497 & 63.1   
& 3525 & 1165.7   
& 109 & 2.9 
& 5493 & 142.8
& 0.1  
& 739  & 5.4   
& ---   & T/O    
& Err & --- 
\\
MontesS10  & 7
& 3643   & 19.1             
& 3643 & 28.3 
& 3643 & 26.6      
& --- & T/O 
& --- & T/O
& T/O       
& --- & T/O     
& ---   & T/O    
& Err & --- 
\\
Wang 93  & 5
& 507    & 44.4              
& 507 & 49.1
& 507 & 46.9                           
& --- & T/O 
& T/O
& --- & 897.1 
& FAIL   & ---     
& Err  & ---    & Err & --- 
\\
Rose$\dagger$  & 3
& 3069     & 200.9               
& 7075 & 498.8 
& 7075 & 477.1                                 
& --- & T/O 
& --- & T/O
& T/O       
& FAIL & ---        
& ---   & T/O    
& Err & --- 
\\
genLinSyst-3-2$\dagger$  & 11
& 222821 & 3087.5 
& --- & T/O                 
& --- & T/O                                    
& FAIL & --- 
& FAIL & --- 
& T/O       
& FAIL   & ---     
& Err & ---    
& Err & --- 
\\
BC-Kahan & 2
& 55     & 0.2                 
& 409 & 2.4         
& 409 & 4.9        
& 55 & 0.2 
& 409 & 2.4
& 0.1  
& 261   & 4.8   
& 409   & 1.5  
& Err & ---  
\\
BC-Arcsin & 2
& 57     & 0.1                 
& 225 & 0.9   
& 225 & 1.9     
& 57 & 0.2 
& 225 & 0.9 
& 0.0  
& 225   & 4.8   
& 225   & 0.7  
& 161 & 2.4 
\\
BC-Sqrt  & 4
& 97     & 0.2                   
& 113 & 0.5 
& 113 & 1.3     
& FAIL & --- 
& 113 & 0.6
& 0.0  
& 105   & 4.7   
& 105   & 0.4  
& 73 & 0.0 
\\
BC-Arctan  & 4
& 211   & 3.5                 
& --- & T/O                
& --- & T/O        
& FAIL & --- 
& --- & T/O
& T/O         
& ---  & T/O      
& Err  & ---    
& --- & T/O 
\\
BC-Arctanh  & 4
& 211    & 3.5                    
& --- & T/O                
& --- & T/O               
& FAIL & --- 
& --- & T/O 
& T/O          
& ---  & T/O      
& Err  & ---    
& --- & T/O 
\\
BC-Phisanbut-1  & 4
& 325    & 0.8                
& 389 & 1.8      
& 389 & 5.8   
& FAIL & --- 
& 389 & 3.6
& 0.1  
& 377   & 4.8   
& 389   & 2.0  
& 217 & 0.2 
\\
BC-Phisanbut-4  & 4
& 543    & 1.6            
& 2007 & 13.6    
& 2065 & 21.5 
& FAIL & --- 
& 51763 & 932.5
& 11.9 
& 51763 & 8.6   
& Err  & ---    
& Err & --- 
\end{tabular}
}
\caption{Comparing our new algorithm to our previous work and competing CAD implementations. }
\label{tab:Results}
\end{sidewaystable}

\end{document}